\documentclass[]{amsart}
\usepackage{latexsym,xy}
\xyoption{all}
\usepackage{graphicx,amscd,amssymb,latexsym,subfigure,verbatim,hyperref,epsfig,supertabular}

\newtheorem{defn0}{Definition}[section]
\newtheorem{prop0}[defn0]{Proposition}
\newtheorem{conj0}[defn0]{Conjecture}
\newtheorem{thm0}[defn0]{Theorem}
\newtheorem{lem0}[defn0]{Lemma}
\newtheorem{corollary0}[defn0]{Corollary}
\newtheorem{example0}[defn0]{Example}
\newtheorem{remark0}[defn0]{Remark}
\newtheorem{que0}[defn0]{Question}

\newenvironment{defn}{\begin{defn0}}{\end{defn0}}
\newenvironment{prop}{\begin{prop0}}{\end{prop0}}

\newenvironment{thm}{\begin{thm0}}{\end{thm0}}

\newenvironment{cor}{\begin{corollary0}}{\end{corollary0}}
\newenvironment{exm}{\begin{example0}\rm}{\end{example0}}

\newcommand{\Z}{{\mathbb{Z}}}

\newcommand{\ZZ}{{\mathbb{Z}}}

\newcommand{\Fqstar}{\F_q^{\,*}}
\newcommand{\F}{{\mathbb F}}

\def\F{\mathbb{F}}

\def\ZZ{\mathbb{Z}}

\def\RR{\mathbb{R}}

\def\Fq{\F_q}
\def\Fqstar{\Fq^\ast}

\def\k.{\mathcal{K}_{\bullet}}

\numberwithin{equation}{section}

\begin{document}
\title[Toric Codes and Finite Geometries]%
{Toric Codes and Finite Geometries} 

\author{John B. Little}
\address{Department of Mathematics and Computer Science\\
College of the Holy Cross, Worcester, MA 01610}
\email{jlittle@holycross.edu}


\begin{abstract}
The theory of affine geometries over the rings $\ZZ/\langle q - 1\rangle$
can be used to understand the properties of toric and generalized toric codes over
$\F_q$.  The standard generator matrices of these codes are produced by evaluating 
collections of monomials in $m$ variables 
at the points of the algebraic torus $(\Fqstar)^m$.  The exponent
vector of such a monomial can be viewed as a point in one of these affine geometries 
and the minimum distance of the resulting code is strongly tied to the lines in the finite 
geometry that contain those points.  We argue that this connection is, in fact, even 
more direct than the connection with the lattice geometry of those exponent vectors 
considered as elements of $\ZZ^2$
or $\RR^2$.  This point of view should be useful both as a way to visualize
properties of these codes and as a guide to heuristic searches
for good codes constructed in this fashion. 
 In particular, we will use these ideas to see a reason why these 
constructions have been so successful over the field $\F_8$, but less
successful in other cases. 

This corrected version (posted on 3/6/2017) adds a necessary hypothesis in Proposition~\ref{index}.
A counterexample to the previous version was pointed out to the author by 
Melda G\"or\"ur and the author thanks her for bringing that to his attention.
Theorem~\ref{binomfact} and Corollary~\ref{distbd} are also reformulated so they 
do not refer to the new more restricted Proposition~\ref{index}.
\end{abstract}

\maketitle

\section{Introduction}
We will consider a particular construction of linear block codes over a 
finite field $\Fq$.  Mathematically, our codes are simply vector subspaces 
$C \subset \Fq^n$ whose elements serve as a set of codewords for representing
information.  This sort of encoding is done to increase the reliability of 
communication over noisy channels and has a number of engineering applications.  
 Our standard reference for basic notions 
and notation in coding theory is \cite{hp}.  As usual, $n$ always denotes
the block length and $k$ denotes the vector space dimension
$\dim_{\Fq} C$, so that the set of codewords contains $q^k$ elements.  
The important parameters of a code are $n,k$ and
a third integer $d$ called the minimum Hamming distance.  For these linear
codes, 
$$d = \min_{x\ne 0 \in C} |\{i \mid x_i \ne 0\}|.$$
If we fix $n$ and $k$, the larger the parameter $d$ is, the larger the error 
detection and error correction capacity of a code is.  

The toric codes studied here are a class of $m$-dimensional cyclic codes 
introduced by J. Hansen in \cite{h1}, \cite{h2}.  (The term ``toric code''
is also used in another context that has no direct connection with this one.)  
Hansen uses the geometry of the projective toric variety corresponding
to a polytope $P$ in $\RR^m$ to describe toric codes, but these may also be understood
in a somewhat more concrete way within the general context of evaluation, or 
functional, codes.

\begin{defn}
Let $P$ be the convex hull of a finite set of integer lattice points, 
contained in $[0,q-2]^m \subset \RR^m$ and let 
$L = \mbox{\rm Span}\{x^e : e \in P\cap \Z^m\}$ be the $\Fq$-linear 
span of the monomials $x^e$ in variables $x_1,\ldots,x_m$ corresponding
to the lattice points $e$ in $P$.  The linear block code 
denoted by $C_P(\Fq)$ is the image of the evaluation mapping on 
the $\Fq$-rational points in the standard $m$-dimensional torus over $\Fq$:
\begin{eqnarray*} \mbox{\rm ev} : L & \to & \Fq^{\,(q-1)^m}\\
                                  g & \mapsto & (g(p) : p \in (\Fqstar)^m).
\end{eqnarray*}
\end{defn}

The condition that $P\subset [0,q-2]^m$ implies that the $x^e$ are linearly 
independent as functions on $(\Fqstar)^m$.
In terms of generator matrices, this construction can also be described
as follows.  Let $\alpha$ be a primitive element for $\Fq$. 
If $f \in \Z^m$ is a vector with $0\le f_i \le q-2$
for all $i$, let $p_f$ denote the point $p_f = (\alpha^{f_1},\ldots,\alpha^{f_m})$ 
in $(\Fqstar)^m$.  If $e = (e_1,\ldots,e_m) \in P \cap \Z^m$, 
write
$$(p_f)^e = (\alpha^{f_1})^{e_1}\cdots (\alpha^{f_m})^{e_m} = \alpha^{\langle f,e\rangle}.$$
Then the standard generator matrix for $C_P(\Fq)$ is the 
$(\dim_{\Fq} L) \times (q - 1)^m$ matrix
$$G = ((p_f)^e),$$
whose rows are indexed by $e \in P \cap \Z^m$, and whose columns are
indexed by $f$ or $p_f \in (\Fqstar)^m$.
We note that if $P$ is the interval $[0,\ell-1] \subset \RR$, then $C_P(\Fq)$ is 
simply the Reed-Solomon code $RS(\ell,q)$.  So toric codes are, in a sense,
higher-dimensional generalizations of Reed-Solomon codes.

In applying these ideas, it has turned out to be worthwhile to generalize
this construction slightly, using arbitrary sets 
$S \subset [0,q-2]^m \subset \RR^m$ instead of the whole
set of lattice points in a convex polytope.  These codes will be
denoted by the analogous notation $C_S(\Fq)$.  If $P = {\rm conv}(S)$, then
the code $C_S(\Fq)$ is a subcode of $C_P(\Fq)$.  In the algebraic geometric
language used by Hansen, the $C_S(\Fq)$ codes can be defined using 
incomplete linear systems $V \subset |\mathcal{O}_{X_P}(D_P)|$, where
$X_P$ is the toric variety determined by $P$ and $D_P$ is the
corresponding divisor class on $X_P$.

The survey \cite{mr} covers most of the work on these codes contained
in \cite{j}, \cite{ls}, \cite{r1}, \cite{r2}, \cite{ss}, and \cite{br}.

Toric codes or generalized toric codes are not all as good as Reed-Solomon codes
from the coding theory perspective, but there are some very good codes first found
by this construction.  For instance, \cite{bk} gives a number of codes
over $\F_8$ found by this method that are better than any previously known examples.  

In many of the works cited above, the main focus has been on identifying conditions
on $P$ or on $S$ that imply results about the minimum distance of the 
corresponding codes using the geometry of $P \cap \ZZ^m$ or $S$ 
\emph{as subsets of the integer lattice} $\ZZ^m \subset \RR^m$. 
In particular, the role of Minkowski sum decompositions of subpolytopes
of $P$ and factorizations of the sections of the corresponding line bundle
on the toric surface $X_P$ has been studied rather intensively in \cite{ls}, 
\cite{ss}, and \cite{l}.  

In this note we will describe and exploit a somewhat different point of view. 
We relate the properties of the $C_S(\Fq)$ codes
to properties of the images of the sets $S$ in the finite 
$m$-dimensional affine ring geometry over $\ZZ/\langle q - 1\rangle$, obtained 
by simply reading the exponent vectors $e$ above as elements of 
$(\ZZ/\langle q - 1\rangle)^m$.  The results here are, in a way, complementary
to those from \cite{l}, where we compared the properties of $C_S(\Fq)$ and the  
related code $C_P(\Fq)$ for $P = {\rm conv}(S)$ and $q$ sufficiently large.
Here the focus will be on the special properties of certain $S$ for specific $q$.

We will concentrate mainly on the case $m = 2$ for
simplicity, although the extension to larger $m$ is essentially immediate.
By itself, this amounts mostly to a relatively 
simple translation of known algebraic facts into another sort of geometric language
with some unusual properties.  However, we will argue that this 
alternative point of view is, if anything, even more natural and direct than
studying toric codes via properties of polytopes and integer lattice vectors
in $\ZZ^2$.  Moreover, this approach 
should prove useful both for visualizing how $d$ is 
determined by the properties of $S$ and hence for heuristic searches for codes with 
good $d$.  

We will recall the known 
properties of ring geometries in \S 2.  We will then apply these properties to 
the study of generalized toric codes in \S 3.   In particular, we will see
a very concrete explanation for why $\F_8$ appears to be a particularly
favorable choice of base field, and for why the construction succeeds 
so well there, yet performs relatively poorly over other fields of comparable 
small size.  Finally, in \S 4, we will offer some more speculative comments
about the potential of this code construction and an indication of which 
other finite fields should have properties analogous to those of $\F_8$.    

\vskip 10pt
\noindent{\bf Acknowledgments.} 
This ideas in this long-gestating paper are indirect 
outgrowths of results obtained in several undergraduate research 
projects carried out under the supervision of the author 
by Ryan Schwarz, Alex Simao, and Lauren Buckley 
at the College of the Holy Cross, and by Alejandro Carbonara, Juan Murillo 
and Abner Ortiz at the MSRI-UP 2009 program.  It is a pleasure to acknowledge 
their contributions to the development of the author's thinking.  Computations 
were done using procedures originally written by D. Joyner in {\tt Magma}, 
a symbolic computation package available at
{\tt http://magma.maths.usyd.edu.au/magma/}.

\section{Finite Ring Geometries}

The properties of the finite affine and projective geometries over a finite field
are very well-known and, of course, form the basis for algebraic geometry over
finite fields and many different sorts of applications to coding theory.
Perhaps less well-known to many mathematicians not working in the area
is that there is also a quite well-developed theory of affine and projective coordinate 
geometries over rings.  We will only need the following relatively simple case discussed
in \cite{le} and called affine \emph{Barbilian planes} there.  These are geometries
with more of the ``usual properties'' one expects from the geometry of the Euclidean
plane than the even more general structures called Hjelmslev planes.

Let $R$ be a ring with multiplicative identity $1$ in which $a\cdot b = 1$ implies 
$b\cdot a = 1$.  Examples include commutative rings with identity as well as 
various noncommutative rings such as matrix rings over a field.  Let $B$ be a subset of 
$R^2 = R \times R$ that satisfies 

\begin{itemize}
\item[$(E_1)$] $(1,0), (0,1) \in B$, 
\item[$(E_2)$] If $(u,v) \in B$ and $r$ is a unit in $R$, then $r(u,v) = (ru,rv) \in B$,
\item[$(E_3)$] Every $(u,v) \in B$ can be completed to an invertible $2\times 2$
matrix $\begin{pmatrix} u & v\\ s & t\end{pmatrix}$ with $(s,t) \in B$, 
\item[$(E_4)$]  If $\begin{pmatrix} u & v\\ s & t\end{pmatrix}$ is an invertible
$2 \times 2$ matrix with $(u,v), (s,t) \in B$, then $(u,v) + \ell (s,t) \in B$
for all $\ell \in R$.  
\end{itemize}

It is easy to see that if $(u,v) \in B$ there must be $s,t \in R$ such that
$su + tv = 1$ and this shows that in the cases we will consider, there is only
one choice for $B$, namely \emph{the set of all $(u,v)$ appearing as rows 
in $2\times 2$ invertible matrices with entries in $R$}.  From now on, 
$B$ will refer to this set; we will not include any indication of the 
ring $R$ in the notation, though, since that should always be clear from
the context. 

One can define 
a geometric structure 
$$\mathcal{G} = (\mathcal{P}, \mathcal{L}, \not\!\circ, \|)$$ 
associated to $R$ as follows:

\begin{itemize}
\item  $\mathcal{P}$, called the set of \emph{points}, is simply $R^2$.
\item  The subsets of $\mathcal{P}$ of the form 
$$(a,b) + R(u,v) = \{(a + \ell u, b + \ell v) \mid \ell \in R\}$$
\emph{with $(u,v) \in B$}
are called \emph{lines} and $\mathcal{L}$ is the collection of all such lines. 
\item  Two points $(a,b)$ and $(c,d)$ are said to be \emph{non-neighbors}, 
written 
$$(a,b) \not\!\circ\,\, (c,d),$$ 
if $(a - c, b - d) \in B$.  If this does not hold, we write $(a,b) \circ (c,d)$
and say that the two points are \emph{neighbors}.
\item  Two lines $\ell_1 = (a,b) + R(u,v)$ and $\ell_2 = (c,d) + R(s,t)$ are said to be 
\emph{parallel} if and only if $R(u,v) = R(s,t)$.  We write $\ell_1 \| \ell_2$
if this is true.  Parallelism is an equivalence relation on the set 
$\mathcal{L}$.  
\end{itemize}
For simplicity, we will call $\mathcal{G}$ the \emph{affine plane over} $R$.

The lines in the affine plane $\mathcal{G}$ have a familiar-looking parametric form and
the points on a line are in one-to-one correspondence with the elements of $R$
because it is required that $(u,v) \in B$.  Two non-neighbor points are
contained in a unique line and parallel lines either coincide or are disjoint.   But 
it is also possible for two distinct neighbor points to be contained in 
\emph{more than one line}, and similarly, it is possible for two non-parallel lines
to intersect in \emph{more than one point}.  

The exact properties of the geometries obtained by this construction are
captured by the list of six axioms from \cite{le} part I, defining the affine
Barbilian planes.  In addition to the properties already mentioned, 
there is a nice analog of the Playfair form of the Euclidean 
Parallel Postulate that holds here.  We will not list all of these properties
because we will not need to make use of them in the following.

On the other hand, \cite{le} part II also contains a number
of results characterizing special properties of these ring geometries
corresponding to some standard ring-theoretic properties of $R$.  For instance, 
we will need the following statements.

\begin{thm}[\cite{le}, part II]
\label{ringprops}  
Let $R$ be a ring with identity with the property 
that $a\cdot b = 1$ implies $b \cdot a = 1$ and let $\mathcal{G}$
be the corresponding affine plane over $R$.  Then
\begin{enumerate}
\item[(1)] The geometry satisfies the analog of Pappus's theorem on triples of points
on two distinct lines if and only if $R$ is commutative
\item[(2)] The neighbor relation on the set of points is transitive if and only if
$R$ is a local ring and $B$ is the set of pairs $(u,v)$ where at least one 
of $u,v$ is a unit in $R$.
\item[(3)] There is at most one line containing any pair of distinct points 
if and only if $R$ has no zero divisors.
\item[(4)] The following are equivalent:
\begin{enumerate}
\item Every pair of distinct points on a line are non-neighbors.
\item Every pair of distinct points are non-neighbors.
\item $R$ is a field (not necessarily commutative).
\item The affine plane satisfies the analog of Desargues' theorem.
\end{enumerate}
\end{enumerate}
\end{thm} 

We will want to make use of this construction in the particular case 
$R = \ZZ/\langle r\rangle$ for some integer $r > 1$.   From Theorem~\ref{ringprops}, 
we easily derive the following statements.

\begin{cor}
Let $R = \ZZ/\langle r\rangle$ and $\mathcal{G}$ be the affine plane 
over $R$.  In Theorem~\ref{ringprops}, in this case:
\begin{itemize}
\item[(a)] The analog of Pappus's theorem in (1) always holds.
\item[(b)] The statements in part (2) hold if and only if $r$ is a \emph{prime power}.  
\item[(c)] The statements in parts (3) and (4) hold if and only if 
$r$ is \emph{prime} and $\mathcal{G}$ is the affine plane over a field.
\end{itemize}
\end{cor}

\begin{exm}
\label{Z8Plane}
Consider the affine plane over $R = \ZZ/\langle 8\rangle$.  We will represent
elements of $R$ by the smallest nonnegative elements of the corresponding
congruence classes.  The set $B$
consists of vectors $(u,v)$ where either $u$ or $v$ is a \emph{unit mod} $8$, hence
equals $1,3,5$, or $7$.  Then for instance $P = (0,0)$ and $Q = (1,4)$ satisfy
$P \not\!\circ\,\, Q$
since $P - Q = (7,4) \in B$.  $P$ and $Q$ are contained in exactly one line:
$$(0,0) + R (1,4) = \{(0,0),(1,4),(2,0),(3,4),(4,0),(5,4),(6,0),(7,4)\}.$$
From this list of points, we can see already that the affine plane over $R$
has some unusual properties.  For example, note that $(2,0)$ also 
lies on the line 
$$(0,0) + R (1,0).$$
So $(0,0) \circ (2,0)$ and these are examples of neighbors lying on 2 distinct lines.  
Similarly the points $(0,0) \circ (4,0)$
are neighbors and they actually both lie on \emph{four} distinct lines:
$$(0,0) + R(1,0), (0,0) + R(1,2), (0,0) + R(1,4), (0,0) + R(1,6).$$  
The set of all neighbors of $(0,0)$ is 
$$\{(a,b) \mid a,b \in \{0,2,4,6\}\}.$$
The neighbor relation is transitive in this case since $R$ is a local ring
with unique maximal ideal $\langle 2\rangle R$ (as in part (2) of Theorem~\ref{ringprops}).
$\diamondsuit$
\end{exm}

The properties seen in this example generalize immediately.

\begin{prop}
\label{index}
Let $r$ be a prime power.  
Let $\mathcal{G}$ be the affine plane over $R = \ZZ/\langle r\rangle$ and 
let $(a,b) \circ (c,d)$ be distinct neighboring points.  The number of distinct
lines containing both points is equal to $r/o((a-c,b-d))$, where 
$o((a-c,b-d))$ is the order of the element $(a-c,b-d)$ in the
additive group $R^2$.   
\end{prop}

\begin{proof}
Without loss of generality, we reduce to the case $(c,d) = (0,0)$.  Let
$(a,b)$ be any point contained in a line $R (s,t)$ and change
coordinates by an invertible $2\times 2$ matrix with entries in $R$ to 
map $(s,t)$ to $(1,0)$, hence mapping $(a,b)$ to 
$\ell (1,0) = (\ell, 0)$ for some $\ell \in R$.
Then for each line containing $(0,0)$ and $(\ell,0)$, 
there is a direction vector $(u,v) \in B$ and $\ell'\in R$ such that 
\begin{equation}
\label{lineeq}
(\ell,0) = \ell'(u,v).
\end{equation}
Note that \eqref{lineeq}  implies $v$ cannot
be a unit in $R$.  Hence by part (2) of Theorem~\ref{ringprops}, 
$u$ must be a unit (since the vector $(u,v) \in B$ 
by definition) and we can replace
$(u,v)$ by another direction vector for the same line having 
the form $(1,v')$.  Then there must be an equation similar
to \eqref{lineeq} with $(u,v)$ replaced by $(1,v')$.  
Then the scalar multiple $\ell' (1,v')$ giving $(\ell,0)$ must have 
$\ell' = \ell$ and $\ell v' = 0$.
Moreover, there is a one-to-one correspondence between the lines containing
$(0,0)$ and $(\ell, 0)$ and solutions of the equation $\ell v' = 0$ in $R$. 
The number of solutions of this equation is equal to the 
index of the subgroup $\langle \ell \rangle \subseteq R$, which
is equal to $r/o(\ell)$.  This establishes the claim.     
\end{proof}

\section{Ring Geometries and Generalized Toric Codes}

We will now consider how the finite ring geometries introduced in the previous
section relate to toric codes.  We again take $m = 2$ for simplicity although
everything extends without difficulty to larger $m$ as well.  The first 
observation is that since we are evaluating the monomials $x^e$ at points
$p_f$ in $(\Fqstar)^2$ (as in the introduction), the fact that primitive elements 
$\alpha$ for $\Fq$ satisfy $\alpha^{q-1} = 1$ shows that 
$$ e = e' \text{ in } (\Z/\langle q - 1\rangle)^2  \Rightarrow (p_f)^e = (p_f)^{e'}.$$

Hence, in a sense, it is probably \emph{even more natural} to consider the exponent 
vectors $e$ used in the evaluation mapping producing a toric surface code or
one of the generalized toric codes $C_S(\Fq)$ with $m = 2$
as elements of the affine plane $\mathcal{G}$ over $\ZZ/\langle q - 1\rangle$ 
rather than as vectors in $\ZZ^2$ or $\RR^2$.   Our first result is a variation on 
the fact noted in Theorem 3.3 of \cite{lschw} that lattice equivalent polytopes
give monomially equivalent toric codes, giving some additional evidence
for this claim.  This statement appears in \cite{cmo}.

\begin{thm}[\cite{cmo}, Theorem 1]
\label{monequiv}
Let $M$ be an invertible $2 \times 2$ matrix with entries in $\ZZ/\langle q-1\rangle$, 
$v$ be a fixed column vector with entries in $\ZZ/\langle q - 1\rangle$, and 
consider the affine mapping 
\begin{align}
T : (\ZZ/\langle q-1\rangle)^2 &\to (\ZZ/\langle q-1\rangle)^2\nonumber\\
              w &\mapsto M w + v\nonumber
\end{align}
Let $S_1$ and $S_2$ be subsets of $(\ZZ/\langle q - 1\rangle)^2$ such that
$S_2 = T(S_1)$.
Then the generalized toric codes $C_{S_1}(\Fq)$ and $C_{S_2}(\Fq)$
are monomially equivalent.  
\end{thm} 

\begin{proof}
The proof is essentially the same as that of Theorem 3.3 from \cite{lschw}.
The component of the vector ${\rm ev}(x^e)$ corresponding to 
$e\in S_1$ and $p_f \in (\Fqstar)^2$
is $\alpha^{\langle e,f\rangle}$.  Similarly, evaluating
$x^{Me + v}$, where $Me + v\in S_2$, we obtain
$$\alpha^{\langle Me + v,f\rangle} = \alpha^{\langle v,f\rangle} \cdot 
\alpha^{\langle e, M^t f\rangle}.$$    
Because it is assumed invertible, $M$ defines a permutation 
of $(\ZZ/\langle q-1\rangle)^2$, and similarly $M^t$ induces a permutation
of $(\Fqstar)^2$. Moreover, the translation vector induces different
constant multiples in each component of the evaluation of a monomial.
Hence the $C_{S_2}(\Fq)$ code is monomially equivalent to the $C_{S_1}(\Fq)$
code.
\end{proof}

The transformations $T$ described here form a group under composition, known
as the affine general linear group over $\ZZ/\langle q - 1\rangle$, and denoted
by ${\rm AGL}(2,\ZZ/\langle q - 1\rangle)$.  
If $S_2 = T(S_1)$ for some such $T$, the sets
are said to be ${\rm AGL}(2,\ZZ/\langle q - 1\rangle)$-equivalent.
Because $\det(M)$ can be any unit in $\ZZ/\langle q - 1\rangle$, not just 
$\pm 1$ as for invertible integer affine transformations, we obtain 
larger equivalence classes here than when we consider lattice equivalence
classes of sets $S$.  But the generalized toric codes for all $S$ in 
one of these equivalence classes are equivalent from the coding theory
perspective -- they have the same total weight enumerators, for instance.   

The following simple algebraic fact will play a key role in relating 
properties of toric codes to the properties of the affine plane over 
$\ZZ/\langle q - 1\rangle$.  The ring $\Fq[x,y]/\langle x^{q-1} - 1, y^{q-1} - 1\rangle$ 
is precisely the coordinate ring of the torus $(\Fqstar)^2$ over $\Fq$.
We will show that the geometry of the affine 
plane over $\ZZ/\langle q - 1 \rangle$ and the algebra of polynomial
functions on the torus are closely connected.  
We will now abandon the multiindex notation and write out monomials in two 
variables explicitly. 

\begin{thm}
\label{binomfact}
Let $\mathcal{G}$ be the affine plane over $R = \ZZ/\langle q - 1\rangle$, let $N$ be 
a factor of $q - 1$, and let $(0,0) \circ (a,b)$ be neighbors with $(a,b) = N(u,v)$ in $R^2$ for 
some $(u,v) \in B$.  Then the binomial $x^a y^b - 1$
factors in $\Fq[x,y]/\langle x^{q-1} - 1, y^{q-1} - 1\rangle$ 
into a product of $N$ distinct factors. 
\end{thm}

Before proceeding to the proof, 
we note that in the case $r$ is a prime power, by Proposition~\ref{index},
the number $N$ is the 
number of lines through $(0,0)$ and $(a,b)$ in $\mathcal{G}$.  In other cases
the number of lines can be even larger.  I would like to thank Melda G\"or\"ur
for bringing examples like this one to my attention.  In the affine plane
over $\ZZ/\langle 6\rangle$, consider the neighbors $(0,0)$ and $(2,0)$.
We have $(2,0) = 2\cdot (1,0) = 2\cdot (1,3) = 2\cdot (4,3)$.  All three
vectors $(1,0),(1,3),(4,3)$ appear in invertible $2\times 2$ matrices over
$\ZZ/\langle 6\rangle$.  Hence they are all elements of $B$.  But this
says there are three distinct lines through these neighbor points, not two.  
On the other hand each of the vectors $(u,v) = (1,0), (1,3), (4,3)$ produces
a factorization of the binomial $x^2 - 1$ in $\F_7[x,y]/\langle x^6 - 1,y^6 - 1\rangle$
as in the statement of the theorem.

\begin{proof}
The integer $N$ is a factor of $q - 1$. Hence $\Fqstar$ contains $N$
distinct $N$th roots of unity and $u^N - 1$ factors completely into 
linear factors in $\Fq[u]$.  But then the same will be true for 
$x^a y^b - 1$ since $(a,b) = N(u,v)$ for some vector $(u,v) \in B$.
If $\alpha$ is a primitive element for $\Fq$, the factorization 
can be written explicitly as
\begin{equation}
\label{factexact}
x^ay^b - 1 = \prod_{\{j \mid (\alpha^j)^N = 1\}} (x^u y^v - \alpha^j).
\end{equation}
This establishes the theorem.
\end{proof}

We are now ready to see some first consequences for toric codes.  

\begin{cor}
\label{distbd} 
Suppose the set $S$ used to produce the generalized toric code
$C_S(\Fq)$ contains $(0,0)$ and $(a,b)$ as in the statement of Theorem~\ref{binomfact}
(or more generally any two elements of $(\ZZ/\langle q - 1\rangle)^2$ that 
differ by $N(u,v)$ for some $(u,v) \in B$).  Then the minimum distance
of $C_S(\Fq)$ satisfies 
$$d(C_S(\Fq)) \le (q - 1)^2 - N(q - 1).$$
\end{cor}

We assume nothing about other points on the lines containing $(0,0)$ and 
$(a,b)$.  Note that this does not contradict the statement of Proposition 3.4
of \cite{lschw} because that article considers only toric codes
from polytopes and we are considering generalized toric codes here.  
The analogous pairs of points here, yielding codes where the bound has $N = 1$,
are those pairs differing by an element of 
order $q - 1$ in $(\ZZ/\langle q - 1\rangle)^2$, or equivalently an 
element of the set $B$.

\begin{proof}  First, if $S$ contains $(0,0)$ and $(a,b)$, then the $\Fq$-span of 
the monomials corresponding to $S$ contains all linear combinations of 
$1$ and $x^ay^b$. Therefore, from \eqref{factexact}, we obtain a codeword 
containing zero entries at positions corresponding to each of the $(x,y) \in (\Fqstar)^2$ with 
$$x^u y^v - \alpha^j = 0$$
as $\alpha^j$ runs through the $N$th roots of unity in $\Fqstar$.  There
are exactly $q - 1$ such points for each $j$.  Moreover the sets of zeroes
are clearly pairwise disjoint.  Hence that codeword has weight $(q - 1)^2 - N(q - 1)$,
and we have an upper bound for $d(C_S(\Fq))$ as claimed.  The more general case
given in parentheses in the statement of the Corollary follows from this.  If
$x^a y^b$ and $x^c y^d$ are in $S$ and $(c - a,d - b) = N(u,v)$ for some $(u,v) \in B$, then 
$$x^c y^d - x^a y^b = x^a y^b (x^{c - a} y^{d - b} - 1).$$
The monomial $x^a y^b$ is nonzero at all points in $(\Fqstar)^2$ and we 
proceed as before with the other factor.
\end{proof}

A direct consequence of this is the following statement about a related
configuration of points.  

\begin{cor}
\label{parallelograms}
Suppose the $S$ used to produce the generalized toric code $C_S(\Fq)$
contains the vertices of a ``parallelogram'' -- that is four points of the form
$$(0,0), (a,b), (c,d), (a+c,b+d),$$
where the sum is taken in $(\ZZ/\langle q - 1\rangle)^2$ (or more generally
something obtained from this by translating by a fixed vector in 
$(\ZZ/\langle q - 1\rangle)^2$).  Assume that $(a,b) = N_1(a',b')$ for $(a',b') \in B$,
and $(c,d) = N_2(c',d')$ for $(c',d') \in B$ in $(\ZZ/\langle q - 1\rangle)^2$.
If $(a',b')$ and $(c',d')$ generate the additive group $(\ZZ/\langle q - 1\rangle)^2$, then 
$$d(C_S(\Fq)) \le (q - 1)^2 - (N_1+N_2)(q - 1) + N_1N_2.$$
\end{cor}

\begin{proof}
Among the linear combinations of the monomials corresponding to the points
in $S$ are combinations that factor as
$$(x^a y^b - 1)(x^c y^b - 1).$$
Because of the hypothesis on $(a',b')$ and $(c',d')$ the curves 
$x^{a'} y^{b'} - \alpha^j = 0$ and $x^{c'}y^{d'} - \alpha^k = 0$ always
intersect in a single point in $(\Fqstar)^2$.  
Applying Theorem~\ref{binomfact} and the proof of Corollary~\ref{distbd},
this polynomial has $(N_1 + N_2)(q - 1) - N_1N_2$ zeroes in $(\Fqstar)^2.$
\end{proof}

We leave it to the reader to formulate and prove a result describing the possibilities
that can occur when $(a',b')$ and $(c',d')$ fail to generate the additive group 
$(\ZZ/\langle q - 1\rangle)^2$.  
In Example~\ref{BadConfig} below, we will see one example of what can happen.

\begin{exm}
\label{Z8bds}
Consider generalized toric codes $C_S(\F_9)$.  Here $q = 9$ so $q - 1 = 8$
and the relevant affine plane is the one whose properties were studied in 
Example~\ref{Z8Plane}.  From Corollary~\ref{distbd}, we obtain, for instance
that if $S$ contains $(0,0)$ and $(4,0)$ (or $(1,0)$ and $(5,0)$, etc.)
then 
$$d(C_S(\F_9)) \le 64 - 4\cdot 8 = 32.$$
From \cite{g}, the best possible $d$ for a code with $n = 64$ and 
$k = 2$ over $\F_9$ has $d = 57$.  Hence such $C_S(\F_9)$ can be far from optimal.
Similarly if $S$ contains any two points differing by an element of order
$N = 2$ in $(\ZZ/\langle 8\rangle)^2$, then 
$$d(C_S(\F_9)) \le 64 - 2\cdot 8 = 48.$$  
We can summarize the pattern here by saying that the presence of 
neighboring points in $S$ tends to reduce $d(C_S(\Fq))$ directly
in proportion to the number of distinct lines through the neighbors.
$\diamondsuit$
\end{exm}

By part (3) of Theorem~\ref{ringprops}, there will be analogous more or
less ``bad'' configurations of pairs or other small numbers of points 
that must be avoided in $S$ in order to produce generalized toric codes over 
$\Fq$ with good minimum distance.  
Here are several examples illustrating these claims.  

\begin{exm}
\label{Collinear}
In \cite{l} we discussed several cases where, even though $S_0$ contains
``gaps'' (that is, if $P = {\rm conv}(S_0)$, then some points of $P\cap \ZZ^2$
are not contained in $S_0$) the generalized code $C_{S_0}(\Fq)$ behaves like a toric code
$C_P(\Fq)$ where $P$ contains a whole line segment and there are linear combinations
of the corresponding monomials that factor completely.  The same kind 
of thing can now be recognized and predicted in many additional examples.  For instance
consider the set $S_0 = \{(0,0),(3,1),(1,3)\}$ with $q = 9$.  Even though
these points are not collinear as elements of $\ZZ^2$, they \emph{are} collinear
in the affine plane over $\ZZ/\langle 8\rangle$ because $(1,3) = 3\cdot (1,3)$.  Moreover, 
if $\beta_i$ are any distinct elements of $\F_9$ with $\beta_1 + \beta_2 + \beta_3 = 0$, then
$$x y^3 + (\beta_1\beta_2 + \beta_1\beta_3 + \beta_2\beta_3) x^3 y + \beta_1\beta_2\beta_3$$
factors as
$$(x^3 y + \beta_1)(x^3 y + \beta_2)(x^3 y + \beta_3)$$
in $\F_9[x,y]/\langle x^8 - 1, y^8 - 1\rangle$.  
This implies that if $S$ contains any set ${\rm AGL}(2,\ZZ/\langle 8\rangle)$-equivalent to 
$S_0$, then generalized toric code satisfies
$$d(C_S(\F_9)) \le 64 - 3\cdot 8 = 40.$$   
The behavior seen in cases like this one, and the similar factorization 
over $\F_8$ from Example 5.6 of \cite{ls}, becomes much less mysterious with the 
viewpoint provided by the finite geometry. $\diamondsuit$
\end{exm}

In the following examples, we will consider codes over $\F_9$ and we will use
a primitive element $\alpha$ for this field given as a root of $u^2 + u + 2 = 0$.

\begin{exm}
\label{Frob}
Consider $S_0 = \{(1,0),(0,1),(3,6)\}$
in the affine plane over $\ZZ/\langle 8\rangle$.  These points are not collinear, but
replacing $x$ by $x^9$ (which defines the same function on $(\F_9^{\,*})^2$) and $y$ by $y^9$, 
we obtain a factorization of a linear combination of $x^9, y^9, x^3 y^6$ as follows:
$$x^9 + y^9 + x^3 y^6 = (x + \alpha y)^3 (x + \alpha^3 y)^3 (x + \alpha^4 y)^3$$
So if $S$ contains any configuration ${\rm AGL}(2,\ZZ/\langle 8\rangle)$-equivalent 
to $S_0$, then 
$$d(C_S(\F_9)) \le 64 - 3\cdot 8 = 40.$$
Whenever $q = p^r$ for $r > 1$, the Frobenius automorphism of the field $\Fq$ will 
produce analogous sets of exponents for which the $C_S$ codes have perhaps unexpectedly 
small minimum distance. 
$\diamondsuit$
\end{exm}

Moreover, those ``bad'' configurations depend strongly on $q$ because the 
geometries of $(\ZZ/\langle q - 1\rangle)^2$ also depend strongly on $q$,
not just on the locations of the points from $S$ in $\ZZ^2$ or $\RR^2$.

\begin{exm}
\label{BadConfig}
Consider the ``trapezoid'' 
$S_0 = \{(0,0),(3,0),(1,4),(2,4)\}$, viewed as a subset of the 
affine planes over $\ZZ/\langle 6\rangle, \ZZ/\langle 7\rangle$, and
$\ZZ/\langle 8\rangle$ in turn.  The corresponding toric codes $C_{S_0}(\Fq)$ have
parameters as follows:
\begin{eqnarray*}
C_{S_0}(\F_7)  && [36,4,18]\\
C_{S_0}(\F_8)  && [49,4,36]\\
C_{S_0}(\F_9)  && [64,4,40].
\end{eqnarray*}
As a result, the presence of $S_0$ (or, by Theorem~\ref{monequiv}, any 
other configuration $S_1$ obtained from $S_0$ by an invertible affine
transformation of the corresponding plane) in a set $S$ imposes 
different ``penalties'' $n - d$ depending on $q$.  The penalty is much 
larger for $q = 7$ or $q = 9$ than it is for $q = 8$.  

The explanation for this behavior comes from the finite geometries.
In the geometry over $\ZZ/\langle 6\rangle$, the points $(0,0)$ and 
$(3,0)$ are neighbors with $(3,0) = 3\cdot (1,0)$ and $(1,0) \in B$. 
(These two points are actually contained in 5 different lines.)  We have
$d(C_{S_0}(\F_7)) \le 36 - 3\cdot 6 = 18$ by considering explicit
factorizations of $x^3 - 1$ as in Theorem~\ref{binomfact}.

In the plane over $\ZZ/\langle 8\rangle$, on the other hand, the 
situation is more subtle.  First, we note that in $(\ZZ/\langle 8\rangle)^2$, 
the configuration $S_0$ is actually also a 
``parallelogram.'' This is true since $(2,4) - (0,0) = (2,4)$
and $(3,0) - (1,4) = (2,4)$.  But we also have 
$(2,4) = 2(1,2)$ and the vectors $(1,4) = (1,4) - (0,0)$ and $(1,2)$ do not 
generate all of $(\ZZ/\langle 8\rangle)^2$.  As a result, the statement of 
Theorem~\ref{parallelograms} does not apply and while the bound is still
true, it is not sharp.  We can understand what is happening in this example algebraically 
by working in $\F_9[x,y]/\langle x^8 - 1, y^8 - 1\rangle$, the coordinate 
ring of the torus $(\F_9^\ast)^2$.  One minimum-weight word in the $C_S(\F_9)$ code comes from 
evaluating
$$
\alpha^7 + \alpha^2 x^3 + \alpha^6 x y^4 + \alpha^3 x^2 y^4
\equiv \alpha^2 (y^2 + x)(\alpha^4 y^2 + x)(\alpha y^4 + x)
$$
(recall that $y^8 \equiv 1$).  This is a maximally
factorizable polynomial in the span of $1, x^3, xy^4, x^2y^4$.
The number of zeroes in $(\F_9^\ast)^2$ turns out to 
be $3\cdot 8 = 24$ in this case, since the curves 
$$y^2 + x = 0, \quad 2y^2 + x = 0, \quad \alpha y^4 + x = 0$$
defined by the factors do not intersect at $\F_9$-rational points in the torus.
$\diamondsuit$
\end{exm}

We believe that the lesson of examples like these is that toric codes over 
fields such as $\F_7$ and $\F_9$ are not automatically bad, but that there 
are certain configurations of points special to the field $\Fq$
that must be avoided in $S$ in order to 
find codes $C_S(\Fq)$ with good $d$.   Here is an example where this
approach was followed to try to find a good code.

\begin{exm}
\label{F9success}
The following $S$ giving a nearly optimal $C_S(\F_9)$ code with parameters
$[64,8,45]$ was found by a randomized heuristic search at the MSRI-UP 2009
undergraduate research program by then-students
Alejandro Carbonara, Juan Murillo, and Abner Ortiz:
$$S = \{(0,4),(1,1),(2,0),(2,3),(2,5),(3,7),(5,2),(7,4)\}.$$
According to \cite{g}, the best known $d$ for this $n$ and $k$ over $\F_9$
is $d = 46$.  It is not difficult to check that all but four of the pairwise
difference vectors $(a,b) - (c,d)$ for $(a,b),(c,d)\in S$ are contained in 
the set $B$ considered here (for the field $\F_9$).  
Moreover the four that are not in $B$, such as $(2,0) - (0,4) \equiv (2,4)$, 
are all elements of order 4 in $(\ZZ/\langle 8\rangle)^2$. 
So the upper bound $d \le 48$ from Corollary~\ref{distbd} or Example~\ref{Z8bds} applies.
This is a case where taking one pair of the points in $S$ gives a code
with $d = 48$, but then adding six more points decreases $d$ by only an additional $3$.

Another observation is that the set of differences $(a,b) - (c,d)$
contains only two pairs of equal vectors (there are $26$ different vectors in 
the set of differences).  The two pairs of equal vectors consist of vectors
in $B$.  Hence there are two ``parallelograms'' contained in 
$S$, and Corollary~\ref{parallelograms} applies with $N_1 = N_2 = 1$.
This gives a less tight upper bound of $d \le (9 - 1)^2 - 2(9 - 1) + 1 = 49$.   
$\diamondsuit$  
\end{exm}

\section{Final Comments}

Computations done by my student Lauren Buckley at Holy Cross in 2014 
show that the value $d = 45$ from Example~\ref{F9success}
is \emph{optimal for generalized toric codes} with
$n = 64$ and $k = 8$ over $\F_9$.
But the method required a detailed (and somewhat tedious) case-by-case analysis
and we will not attempt to present
the details here.  The idea was to enumerate all the 
${\rm AGL}(2,\ZZ/\langle 8\rangle)$-equivalence classes of base sets $S_0$ with 
$|S_0| = 4$, and then consider all possible
ways to ``build up'' to $k = 8$ by adding 4 additional points to one $S_0$
in each class.  As $k$ increased, it quickly became impossible to avoid 
some sets dropping $d$ to $45$ or less.  The examples presented
above were all used to recognize when this happened. Needless to say, 
though, we would like to have a better argument to show $d \le 45$ in this case.

We will conclude this note by making some further observations regarding
the potential of the generalized toric code construction for producing
really good codes (say better than those found by other methods and recorded
in the database \cite{g}).  As we mentioned previously, this construction
has been most successful over $\F_8$, as shown for example in the 
new codes found in \cite{bk}.  The underlying reason for this should be
somewhat clear by now -- we believe that this is simply a reflection of
the fact that the underlying geometry in this case comes from the 
the ring $R = \ZZ/\langle 8 - 1\rangle$ which is isomorphic to the field $\F_7$.  
Since this $R$ has no zero divisors, there are no neighboring points in the affine 
plane $R^2$.  All of the 
properties in (4) of Theorem~\ref{ringprops} hold in this case, so 
there are many fewer ``bad configurations'' to avoid in searches
for good codes.  

While there are isolated examples like the one in Example~\ref{F9success}
over fields $\Fq$ for which $\ZZ/\langle q - 1\rangle$ is not also a field, and
even a few others where optimal codes have been obtained as generalized
toric codes, we believe that these cases will be much rarer and more difficult
to find.  

The best next case to look at will probably be codes over
$\F_{32}$ and more generally the other cases where $p$ is a Mersenne 
prime and $p + 1 = 2^r$.  But of course those cases are relatively rare
and they lead to large fields where virtually nothing is known yet
about optimal codes.  

\bibliographystyle{amsplain}

\end{document}